\documentclass[abstract=true, DIV=14]{scrartcl}

\usepackage[noEnd,commentColor=black]{algpseudocodex}
\usepackage{amsmath}
\usepackage{amssymb}
\usepackage{amsthm}
\usepackage{ dsfont }
\usepackage[mathscr]{euscript}
\usepackage[shortlabels]{enumitem}
\usepackage{graphicx} % Required for inserting images
\usepackage{subcaption}
\usepackage{thmtools}
\usepackage{todonotes}
\usepackage{xcolor}
\usepackage{microtype}
\usepackage{cancel}
\usepackage{stackengine}
\usepackage[labelfont=bf]{caption}

\makeatletter
\def\@fnsymbol#1{\ensuremath{\ifcase#1\or \!\;\or \!\;\or \ddagger\or
   \mathsection\or \mathparagraph\or \|\or **\or \dagger\dagger
   \or \ddagger\ddagger \else\@ctrerr\fi}}
\makeatother

\usepackage[normalem]{ulem}

\usepackage[linesnumbered,vlined]{algorithm2e}
\usepackage{xcolor}
\usepackage[many]{tcolorbox}

\SetAlgorithmName{Algorithm}{Algorithm}{List of Algorithms}
\SetAlCapFnt{\small}
\SetAlCapNameFnt{\small}
\SetAlgoCaptionSeparator{:}
\SetKwComment{Comment}{$\triangleright$}{}
\SetKwInput{KwInput}{Input}
\SetKwInput{KwOutput}{Output}
\SetKwProg{Procedure}{Procedure}{}{}
\SetKwProg{Function}{Function}{}{}

\colorlet{keywordcolor}{blue!70!black}
\colorlet{commentcolor}{green!50!black}

\SetKw{Let}{let}
\SetKw{KwTo}{to}
\SetKw{KwRet}{return}
\SetKw{Break}{break}
\SetKw{Continue}{continue}

\SetCommentSty{textcolor{commentcolor}}

\SetCommentSty{mycommfont}

\newcounter{customAlgorithm}[section]
\renewcommand{\thecustomAlgorithm}{\thesection.\arabic{customAlgorithm}}

\makeatletter
\@addtoreset{customAlgorithm}{section}
\makeatother

\usepackage{upgreek}
\usepackage[colorlinks]{hyperref}
\usepackage[capitalize]{cleveref}

\crefname{equation}{eq.}{eqs.} % Don't always capitalize
\crefname{enumi}{}{} % No label for items (default)
\crefname{icase}{case}{cases}
\crefname{ipart}{part}{parts}
\crefname{iprop}{property}{properties}
\crefname{iinv}{invariant}{invariants}

\crefformat{section}{\S\,#2#1#3}
\crefformat{subsection}{\S\,#2#1#3}
\crefrangeformat{section}{\S\S\,#3#1#4--#5#2#6}
\crefmultiformat{section}{\S\S\,#2#1#3}{,~#2#1#3}{,~#2#1#3}{,~#2#1#3}

\usepackage{authblk}
\usepackage[normalem]{ulem}

\usepackage{tikz}
\usetikzlibrary{calc}

\DeclareMathOperator{\OPT}{\mathsf{OPT}}
\DeclareMathOperator{\ALG}{\mathsf{ALG}}

\renewcommand{\alpha}{\upalpha}

\newcommand{\cR}{\mathscr{R}}

\newcommand{\connected}[1]{\def\temp{#1}\ifx\temp\empty\sim\else\overset{#1}{\sim}\fi}

\tikzset{
	point/.style={circle, fill, inner sep=1.5pt},
	smallpoint/.style={point, inner sep=1.2pt},
	tinypoint/.style={point, inner sep=1pt},
	hlbox/.style={fill, {white!90!black}},
	subrect/.style={draw, fill={white!80!cyan}}, % Sub-rectangle in some decomposition
	msrect/.style=subrect % Rectangle in merge sequence
}

\newtheorem{theorem}{Theorem}[section]
\newtheorem{lemma}[theorem]{Lemma}

\newtheorem{claim}[theorem]{Claim}

\theoremstyle{definition}

\author{L\'{a}szl\'{o} Kozma\thanks{$^{*}$ Email: \href{mailto:laszlo.kozma@tu-dresden.de}{laszlo.kozma@tu-dresden.de}}} 
\affil{Faculty of Computer Science, TU Dresden, Germany}
\title{The price of incrementality in $k$-center clustering}

\date{}

\begin{document}

\maketitle

\vspace{-0.3in}
\begin{abstract}
The $k$-center problem is one of the  best-studied and most intuitive clustering formulations. It asks, given a set of $n$ points in a metric space, for $k$ of the points to be designated as cluster centers, so that the maximum distance of an input point to its nearest center is minimized. Gonzalez's greedy algorithm from 1985 is a simple and efficient way to find a $2$-approximate solution. The algorithm has the attractive feature of \emph{incrementality}: it outputs the centers one by one, with a guaranteed $2$-approximation for every prefix of the obtained sequence of centers. 

Incrementality imposes a geometric constraint on how solutions can be built, and it is natural to ask whether this comes at a price in the quality of the solution. It is known that in polynomial time, the approximation ratio of $2$ is best possible, assuming $P \neq NP$. In this paper we show that even with \emph{unlimited} computational power, the factor $2$ cannot be improved, if the solution is required to be built incrementally.  The lower bound construction imposes a tradeoff between all $n$ levels of the clustering simultaneously; it was obtained with the help of ChatGPT, an aspect we discuss in \S\,\ref{sec3}.

\end{abstract}

\section{Introduction}\label{sec1}

In a typical formulation of the clustering problem, we are given a set $S$ of $n$ input points with a metric distance function $d$, and a positive integer $k$. The goal is to select a subset $C \subseteq S$ of $k$ cluster centers, minimizing the maximum distance of any point to its nearest center, i.e., the quantity $$\mathrm{cost} = \max_{x \in S}\min_{c \in C} d(x,c).$$

This problem is known as (\emph{discrete})\emph{ $k$-center}. A closely related variant of the problem is \emph{continuous $k$-center}, where the set $C$ of centers is selected from the ambient metric space (typically Euclidean), not necessarily from $S$. 

A natural consideration in many applications is that the value $k$ may not be known upfront, or may be subject to unexpected budget changes. (Think, e.g., of a city deciding on locations for hospital sites in a volatile political environment.) We accordingly formulate the problem in a partially \emph{online} setting: An algorithm is required to commit to centers one-by-one irrevocably, until an adversary announces the end of the process. We call this the \emph{incremental} $k$-center problem. 

The task is then, as is typical for online algorithms, to obtain a good approximation of the theoretical optimum (constructed, in this case, with full knowledge of $k$). More precisely, we denote by $\OPT_k$ the cost of the optimal solution with $k$ centers, and by $\ALG_k$ the cost of the solution found by the algorithm (recall that the algorithm must construct this solution by adding one more center to its solution with $k-1$ centers, whereas the optimum has no such restriction). The goal is to minimize the \emph{ratio} $$\cR = \max_{1 \leq k \leq n}{\frac{\ALG_k}{\OPT_k}},$$ for the worst-case input point set $S$. We emphasize that the algorithm has full offline access to the input point set $S$, and the only information withheld from it is the value $k$, which is revealed at termination. The ratio $\cR$ can be seen as analogous to the competitive ratio (i.e., the price of not knowing $k$), but we will simply refer to it as approximation ratio, understanding that it measures the worst-case performance of incremental algorithms.

Perhaps the most natural algorithm for the $k$-center problem is the greedy farthest-point algorithm due to Gonzalez~\cite{Gonzalez85}. It starts with an arbitrary first center, and then repeatedly selects the next center that is farthest from the current set of centers. Precisely, if $C$ is the current set of centers, then the next center chosen is $x \in S \setminus C$ for which $\min_{c \in C}{d(x,c)}$ is maximized. 
(One may note that the first center could also be picked in a more deliberate way; this, however, does not affect the worst-case guarantees of the algorithm.)  

A simple calculation based on the triangle inequality shows that Gonzalez's algorithm is a $2$-approximation in any metric space, i.e., the cost of its solution is at most twice the optimum. More strongly, its guarantee holds at every step of the algorithm, so even if $k$ is chosen adversarially, the algorithm maintains the approximation with respect to the optimal solution for that value of $k$. 
We refer to this feature of the algorithm as \emph{incrementality}.
In our earlier notation, Gonzalez's algorithm achieves the approximation ratio $\cR \leq 2$, which is easily seen to be tight.

Recall that in our setting, the input point set $S$ is entirely static, and revealed to the algorithm upfront. Dynamic variants of the clustering problem have also been studied, where points can be added or removed, or the input undergoes changes in different ways, e.g., see~\cite{CharikarCFM04, CrucianiFGNS24}. Occasionally, the case when points can only be added has also been called ``incremental clustering''. We hope that confusion can be avoided, and stick to the term \emph{incremental}, to maintain correspondence to the field of \emph{incremental optimization}, where the term is similarly used in the sense of building a solution step-by-step, e.g., see~\cite{BernsteinD0H22}.

\medskip

To summarize, Gonzalez's algorithm has three attractive features: polynomial running time, incrementality, and $2$-approximation. 
It is known that an approximation factor $c<2$ cannot be obtained in polynomial time (unless $P = NP$), as shown via a reduction from dominating set~\cite{HochbaumS85}. 
A natural and perhaps overlooked question is whether the factor $2$ can be improved, while preserving incrementality, without any restriction on the running time, i.e., allowing unlimited computational power. (Note that without requiring incrementality, the discrete $k$-center optimum can clearly be found in exponential time.) This is the question addressed in this paper:

\smallskip 

\textbf{\emph{Is there an incremental $k$-center algorithm with approximation ratio $\cR < 2$?}}

\smallskip

Since we place no restriction on computation, an incremental algorithm can be assumed to have access to the optimal center sets for each choice of $k$. The solutions for different $k$, however, generally do not ``nest'', and cannot be built incrementally from one another. An incremental algorithm thus needs to balance the approximation ratio across all levels of $k$, and being suboptimal at one level may help at another. We think of the approximation ratio $\cR$ as the penalty for having to build the solution step-by-step, not knowing the true value of $k$; in short, \emph{the price of incrementality}. 

In this paper we show that the answer to the above question is negative, and the price of incrementality for $k$-center is $2$. 

\begin{theorem}\label{thm1}
No incremental $k$-center algorithm can have approximation ratio $\cR < 2$.
\end{theorem}

As discussed, a \emph{computational barrier} to improving the factor 2 was already known. Theorem~\ref{thm1} establishes an entirely different, \emph{structural/geometric barrier} to such an improvement. Strikingly, this barrier already holds for input from a very simple metric space: the line. We see no obvious a priori reason why the two lower bounds had to coincide, and indeed, they arise through unrelated arguments. 

\paragraph{Discrete and continuous $k$-center.} We only focus on the discrete $k$-center problem, since an improvement to the factor $2$ is easily ruled out in the continuous variant of the problem. 

Indeed, take an input with just two points. Unless the algorithm chooses exactly these points as centers, its approximation ratio will be infinite in the case $k=2$. On the other hand, choosing one of the points as the first center already implies a factor-$2$ penalty for $k=1$, compared to the optimal choice of the midpoint of the two input points. Thus, any algorithm must incur a penalty factor of at least $2$. Meanwhile, Gonzalez's algorithm does achieve a $2$-approximation for the continuous version of the problem as well.

The $k$-center problem has also been studied in the $d$-dimensional Euclidean space. The problem is known to be NP-hard already for $d=2$ or for $k=2$, when the other parameter is part of the input~\cite{FowlerPT81, Megiddo90}. On the line, $k$-center can be solved in $O(n)$ time~\cite{Frederickson91} (apart from sorting), making it all the more surprising that the incremental hardness result holds already on the line. 

It is not known whether the NP-hardness of approximation beyond factor 2~\cite{HochbaumS85} extends to the Euclidean plane. In this special case, only an inapproximability below a factor $\approx 1.8$~\cite{Mentzer1, Mentzer2} has been established.

\paragraph{Related work.}

Optimization in an incremental setting similar to ours has been studied more broadly, see~\cite{BernsteinD0H22} and references therein; we note that tight results are scarce in this area. 
Incremental optimization has been considered for clustering and related facility location problems~\cite{ArulselvanMS15, LinNRW06, MettuP03, Plaxton06, DuXZ15}; however, these works relate to different clustering costs (e.g., $k$-median or $k$-means), or impose other restrictions on the solution, which make the results incomparable to ours. 

The hardness of computing a $2$-approximate incremental $k$-center solution \emph{in polynomial time} is noted in~\cite{Plaxton06}; as this follows immediately from~\cite{HochbaumS85}, it is not relevant to our study. 

%None of the papers studies however the incrementality of $k$-center, i.e., the question addressed in our papers

The requirement of arranging clusters \emph{hierarchically} imposes a geometric constraint similar in spirit to incrementality, although the results are not directly comparable. For the effect of imposing a hierarchy requirement in clustering we refer to~\cite{ArutyunovaR25} and references therein.

%(We note in passing, that this holds in general metrics, and the optimal approximation ratio, e.g., in $\mathbb{R}^2$ has not been precisely determined.)

\section{The lower bound}

We begin with a warmup construction that shows a lower bound of $\cR \geq 1.618$ (the golden ratio). We then show an improvement to $\cR \geq 1.8$, finally giving a tight lower bound of 2, in the limit.

\subsection{Balancing two levels}\label{sec21}

Let the input consist of five points on the line. Let $\phi = \frac{\sqrt{5}+1}{2} \approx 1.618$, the golden ratio, and place the points, denoted $p_1$ to $p_5$ in a left-to-right order, at positions $$\Big(0,~\phi-1,~\phi,~\phi+1,~2\phi\Big),$$ as shown in Figure~\ref{fig1}.

The optimum is easily seen to be $C = \{p_3\}$ for $k=1$ and $C = \{p_2, p_4\}$ for $k=2$, with costs $\OPT_1 = \phi$, and $\OPT_2 = 1$, respectively.

The choice of the first center already forces a tradeoff for an incremental algorithm. A choice other than $p_3$ results in a cost $\ALG_1 \geq \phi+1$ (as the farthest endpoint is then at distance at least $\phi+1$), for a ratio $$\cR \geq \frac{\ALG_1}{\OPT_1} \geq \frac{\phi+1}{\phi} = \phi > 1.618.$$

Choosing $p_3$ as the first center, the algorithm must take the second center on one side of $p_3$, so the nearest center to the opposite endpoint ($p_1$ or $p_5$) remains $p_3$ itself, at distance $\phi$. This yields $\ALG_2 = \phi$, for a ratio $$\cR \geq \frac{\ALG_2}{\OPT_2} \geq \frac{\phi}{1} > 1.618.$$

Thus, in all cases we have an approximation ratio of at least $\phi$ and no incremental algorithm can achieve $\cR < \phi$. 

\begin{figure} 
\centering\includegraphics[scale=0.35]{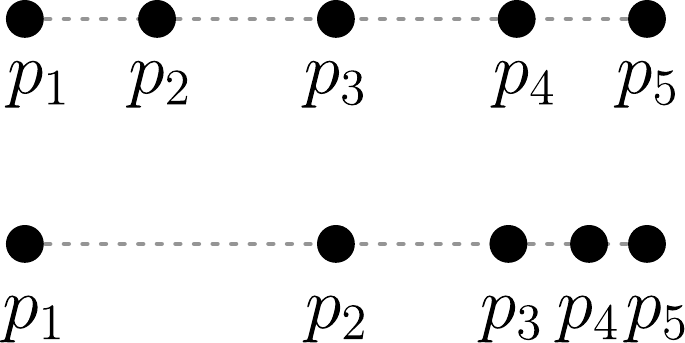}
\captionsetup{width=.92\linewidth}

\caption{Constructions of \S\,\ref{sec21} and \S\,\ref{sec22}, identically scaled.\label{fig1} }
\end{figure}

\subsection{Beyond two levels}\label{sec22}

The previous construction results from optimizing the distances such as to obtain the same competitive ratio at the levels $k=1$ and $k=2$ simultaneously. For broad classes of inputs (even in arbitrary metrics) one can show that the lower bound cannot be strengthened focusing only on these two levels. The key to the improvement is thus to consider more than two levels at once.

Let the input again consist of five points on the line, denoted $p_1$ to $p_5$. Let $\rho =2\cos(\frac{\pi}{7}) \approx 1.802$ be the root of the equation $x^3-x^2-2x+1=0$ in the interval $(1,2)$, and place the points at positions $$\Big(0, ~\rho, ~1 + \rho, ~{\rho^2}, ~2{\rho}  \Big),$$ as shown in Figure~\ref{fig1}.
The approximate values, scaled so that the largest coordinate equals $1$, are $\Big(0,~0.5,~0.7775,~0.9,~1 \Big).$

The optimum solutions for $k=1,2,3$ are easily seen to be $\{p_2\}$, $\{p_1,p_3\}$, and $\{p_1,p_2,p_4\}$, with costs $\OPT_1 = 0.5$, $\OPT_2 = 0.2775$, and $\OPT_3 = 0.1225$, where distances are computed with the normalized approximate values. 

An incremental algorithm has various possibilities. Choosing $p_1$ or $p_5$ as the first center already yields $\cR \geq \ALG_1/\OPT_1 = 2$, we can thus ignore these cases.

Choosing $p_4$ as the first center, we have $$\cR \geq \frac{\ALG_1}{\OPT_1} = \frac{0.9}{0.5} = 1.8,$$ which is the bound we are aiming for, so we can also conclude this case.

Choosing $p_2$ as the first center, we have $\ALG_1 = \OPT_1$, achieving optimality for $k=1$. Whichever point is chosen as the second center, either $p_1$ or $p_5$ will be at distance $0.5$ from its nearest center, resulting in $\ALG_2 = 0.5$, for a ratio $$\cR \geq \frac{\ALG_2}{\OPT_2} = \frac{0.5}{0.2775} > 1.801.$$ 

It remains to consider the case where $p_3$ is the first center. Now, $\ALG_1 = 0.7775$, for a ratio $0.7775 / 0.5 = 1.555$, still better than our target.

Choosing $p_2$, $p_4$ or $p_5$ as the second center leads to a cost $\ALG_2$ of $0.5$, $0.7775$, or $0.7775$, respectively. In either case, we have $$\cR \geq \frac{\ALG_2}{\OPT_2} \geq \frac{0.5}{0.2775} > 1.801.$$

It remains to consider $p_1$ as the second center. The solution with centers $\{p_1,p_3\}$ is optimal for $k=2$. Now any of the choices for the third center leads to a cost of $\min\{p_3-p_2, p_4-p_3\}  = 0.2225$, for a ratio $$\cR \geq \frac{\ALG_3}{\OPT_3} \geq \frac{0.2225}{0.1225} > 1.81.$$

Since we have, in all cases, a ratio of at least $1.8$, we conclude the analysis. Observe that there is a slight slack in the approximate numerical bounds; the sharper bound of $\cR \approx 1.802$ can be obtained by using the exact algebraic expressions given above.

\subsection{All levels at once}\label{sec23}

Fix now an integer $m \geq 3$, and let $\rho = \rho_m$ be the unique solution in the interval $(1,2)$ of the equation $$ \left( \frac{\rho-1}{\rho} \right)^{m-2} =~ 2 - \rho.$$

The fact that there is a unique solution in this interval follows from the observation that both sides are continuous and strictly monotone in the opposite way, with the left side smaller than the right side for $\rho=1$ and greater for $\rho = 2$.

Since the base of the left side is in $(0,1/2)$ in the range of interest, we have $2 - \rho_m \leq 2^{-m+2}$, and we can conclude that $\rho_m$ converges to $2$ rapidly, as $m$ increases.

Let us now choose an input to $k$-center with $m+1$ points $p_0, \dots, p_m$, where the first point is  $p_0 = 0$, and the gap between neighboring points $d_i = p_i - p_{i-1}$ is defined as $d_1 = \rho$, $d_{m} = \rho(2-\rho)$, and $$d_i = \left(\frac{\rho-1}{\rho}\right)^{i-2} \mathrm{~for~} 2 \leq i \leq m-1.$$     

It is interesting to observe that for $m=4$ we recover the example from \S\,\ref{sec22}. The case $m=3$ recovers only part of the warmup example from \S\,\ref{sec21}, an asymmetric construction with only the points $p_1,p_3,p_4,p_5$; a simple case-analysis shows that this already forces the ratio $\cR \geq \phi$.

We state a crucial identity for the tail sum of the gaps:

\begin{claim}
For $2 \leq i \leq m-1$, we have $\displaystyle\sum_{h=i}^m{d_h} = \rho \cdot d_i$.
\end{claim}
\begin{proof}
We proceed by reverse induction on $i$. For $i=m-1$, noticing $d_{m-1} = (2-\rho)\rho/(\rho-1)$, we have $d_{m-1} + d_m = d_{m-1} + d_{m-1}(\rho-1) = \rho \cdot d_{m-1},$ as needed. 

Suppose that the claim holds for $i+1$. Now, since $d_{i+1} = d_{i} \cdot {(\rho-1)}/\rho$, we have $$\quad \quad \quad \quad  \quad \quad \quad \quad \quad \quad \quad  \sum_{h = i}^{m}{d_h} = d_i + \rho \cdot d_{i+1} = d_i + (\rho-1)\cdot d_i = \rho \cdot d_i.\quad \quad \quad \quad \quad \quad \quad \quad \quad \qedhere$$ 
\end{proof}

Observe that from the claim it follows that the total length (the distance from $p_0$ to $p_m$) is $$p_m - p_0 = d_1 + \sum_{h=2}^m{d_h} = \rho + \rho \cdot d_2 =  2\rho.$$

We further note that the sequence of gaps $d_1, \dots, d_m$ is strictly decreasing. This is immediate for $d_2, \dots, d_{m-1}$. To see $d_1 > d_2$, we need to verify $\rho > 1$, and to see $d_{m-1} > d_m$, we need to verify $\rho/(\rho-1) > \rho$, both of which are immediate.  

Let us now compute the optimum cost of this instance.

\begin{lemma}
The optimum $k$-center cost for the input $\{p_0,\dots,p_m\}$ defined above is $\OPT_k = d_k$, for all $k \leq m$.
\end{lemma}

\begin{proof}
We start with the lower bound $\OPT_k \geq d_k$. The first $k$ points $\{p_0, \dots, p_{k-1}\}$ are separated by gaps of at least $d_k$ from their neighbors, so if one of these points is not picked as a center, we incur a cost of at least $d_k$.
Assume therefore that the first $k$ points are picked as centers. Then the distance of the last point $p_{m}$ to its nearest center $p_{k-1}$ is $\sum_{h=k}^{m}{d_h} = \rho \cdot d_k > d_k$.

To show that an optimal cost of $d_k$ can be attained, consider the solution $\{p_1\}$ for $k=1$, $\{p_0, p_1, \dots, p_{k-2}, p_k\}$ for $2 \leq k \leq m-1$, and $\{p_0, p_1, \dots, p_{m-1}\}$ for $k=m$.

In the first case, the cost is the larger of $p_1 - p_0 = d_1$ and the right tail $p_m-p_1$, which equals $\rho \cdot d_2 = d_1$.

In the middle case, the cost is the larger of $d_{k}$ (to cover $p_{k-1}$) and the right tail $p_m-p_{k}$,  which equals $\rho \cdot d_{k+1} = (\rho-1)d_k < d_k$.

In the last case, the cost is $p_m-p_{m-1}=d_m$, as needed.

The lower and upper bounds together imply that $\OPT_k = d_k$ for all values of $k$.
\end{proof}

It remains to bound the cost of an incremental algorithm at different levels of $k$. One may already observe that the optimal solution of the previous lemma cannot be built incrementally. We make this claim more strongly as follows. 

\begin{lemma}
No incremental algorithm can attain cost $\ALG_k < \rho \cdot d_k$ simultaneously for all values $k = 1,\dots,m-1$. 
\end{lemma}

The two lemmas together imply Theorem~\ref{thm1}, since $\cR \geq \rho_m$ and $\rho_m \rightarrow 2$ as $m \rightarrow \infty$.

\begin{proof}
For $k=1$, the chosen center must be one of the points $p_1, p_2, \dots, p_{m-2}$. 

Indeed, choosing the endpoints $p_0$ or $p_m$ would be too far from the opposite endpoint (distance $2\rho = 2\OPT_1$).
Choosing $p_{m-1}$ would lead to cost $\ALG_1 \geq 2\rho - d_m = 2\rho - \rho(2-\rho) = \rho^2$. This equals $\rho \cdot d_1 = \rho \cdot \OPT_1$ with approximation ratio equal to $\rho$, not strictly smaller.

Now consider any $k$ with $2 \leq k \leq m-1$. A set of $k$ centers with cost strictly less than $$\rho \cdot \OPT_k = \rho \cdot d_k = \sum_{h=k}^m{d_h}$$ must be of the form $\{p_0, p_1, \dots, p_{k-2},p_j\}$, for some $j \geq k$. This is because omitting any of the points $p_0, \dots, p_{k-2}$ would lead to a cost $\ALG_k$ of at least $d_{k-1}$, and therefore, $$\frac{\ALG_k}{\OPT_k} \geq \frac{d_{k-1}}{d_k} = \frac{\rho}{\rho-1} > \rho,$$
with an approximation ratio larger than allowed. 

Moreover, the remaining center must be from $\{p_{k-1}, \dots, p_m\}$, but it cannot be either of the two extreme points. Choosing $p_{k-1}$ would lead to a cost $\ALG_k \geq p_m - p_{k-1} = \rho \cdot d_{k}$, and choosing $p_m$ would lead to a cost $\ALG_k \geq p_{k-1} - p_{k-2} = d_{k-1} > \rho \cdot d_k$, in both cases with a ratio at or above the forbidden threshold $\rho$. 

Let the first center chosen by the algorithm be $p_j$, where $1 \leq j \leq m-2$ must hold from our earlier discussion. Consider the solution of the algorithm for $k = j+1$. As we have shown, the solution must include $p_0, p_1, \dots, p_{j-1}$, and must omit $p_j$. But since $p_j$ was already chosen in the first step (and no recourse is allowed), it cannot be omitted, a contradiction. This concludes the proof. 
\end{proof}

\section{Conclusion and the use of AI} \label{sec3}

The question addressed in this paper has occupied the author for some time. Having obtained the warmup construction of \S\,\ref{sec21}, and discussing it with several students and colleagues, opinions were fairly evenly split between those expecting the true bound to be $2$ and those favoring $\phi$ or some intermediate value.

The results described in \S\,\ref{sec22} and \ref{sec23} (i.e., the $\approx 1.8$ construction and its generalization) were found by ChatGPT. The AI model found both constructions from a single prompt, without further interaction or guidance. The prompt described the question and the warmup construction at a level similar to communicating with a competent colleague.
Access to the model and assistance in using it was kindly provided by Alexandr Andoni, which is gratefully acknowledged.

\medskip

The arguably elegant and (to our knowledge) original solution seems, in hindsight, quite accessible. The writeup is the author's own, verifying and expanding on proof steps and motivating some choices in the argument. Yet, it should be emphasized that the output of the AI model was already correct and clearly written, if somewhat compact. 

Several factors and biases may have made it difficult for the author to find the solution earlier. First, the lower bound of $\phi$ seemed natural and plausibly optimal, suggesting to look for an improved algorithm instead. The known construction was symmetric; the asymmetric, four-point version of the warmup construction perhaps more readily leads one towards the generalized geometrically decreasing sequence. Finally, the improvement requires handling more than two $k$-values simultaneously, which (incorrectly) seemed to bring no gains. AI systems may be particularly well-suited to finding lower-bound constructions of this kind.

The author also thanks Heiko Röglin for earlier insightful comments on the problem.

\small
\bibliographystyle{alphaurl}
\bibliography{article}

\end{document}